\documentclass[letterpaper,11pt]{article}
\usepackage[in]{fullpage}

\usepackage{amsfonts,amsthm,amsmath}
\usepackage{xcolor}
\usepackage[ruled,vlined,linesnumbered,noalgohanging,nofillcomment]{algorithm2e}
\SetAlgoHangIndent{0pt}
\DontPrintSemicolon

\usepackage{hyperref}
\hypersetup{
  colorlinks   = true, 
  urlcolor     = blue, 
  linkcolor    = blue, 
  citecolor   = blue 
}

\theoremstyle{plain}
\newtheorem{theorem}{Theorem}
\newtheorem{lemma}[theorem]{Lemma}
\newtheorem{fact}[theorem]{Fact}

\theoremstyle{definition}
\newtheorem{definition}[theorem]{Definition}

\newcommand{\eps}{\epsilon}
\newcommand{\Z}{\mathbb Z}
\newcommand{\N}{\mathbb N}
\newcommand{\R}{\mathbb R}
\newcommand{\II}{\mathcal I}
\renewcommand{\AA}{\mathcal A}
\newcommand{\DD}{\mathcal D}

\newcommand{\zo}{\set{0,1}}

\newcommand{\Asparse}{\mathcal A_\text{sparse}}
\newcommand{\Adensity}{\mathcal A_\text{density}}
\newcommand{\Aapprox}{\mathcal A_\text{approx}}
\newcommand{\Mexact}{M_\text{exact}}
\newcommand{\Mapprox}{M_\text{approx}}
\newcommand{\kapprox}{k_\text{approx}}
\newcommand{\sparse}{\textsc{sparse}}
\newcommand{\dense}{\textsc{dense}}

\newcommand{\tO}{\tilde{O}}
\newcommand{\AdversaryPlayer}{\xspace{\textsf{Adversary}}\xspace}
\newcommand{\AlgorithmPlayer}{\xspace{\textsf{Algorithm}}\xspace}

\renewcommand{\th}{^{\textrm{th}}}

\newcommand{\Ism}{\II_{\text{small}}}
\newcommand{\Ilg}{\II_{\text{large}}}

\newcommand{\vcount}{\mbox{\sl count}}
\newcommand{\vregime}{\mbox{\sl regime}}
\newcommand{\vinterval}{\mbox{\sl interval}}

\DeclareMathOperator{\poly}{poly}
\DeclareMathOperator{\polylog}{polylog}
\DeclareMathOperator{\loglog}{loglog}
\DeclareMathOperator*{\E}{{\mathbb E}}

\newcommand{\moment}[2]{\left\|#2\right\|_{#1}^{#1}}
\newcommand{\pmoment}[1]{\moment{p}{#1}}
\newcommand{\eqdef}{\stackrel{\text{\rm def}}{=}}
\newcommand{\set}[1]{\left\{{#1}\right\}}

\newcommand{\algheader}[1]{\mbox{\textbf{#1:}}}
\newcommand{\noheader}[1]{\mbox{\phantom{\textbf{#1:}}}}

\newcommand{\aprx}[1][]{$(1\pm\alpha#1)$-approximation\xspace}
\newcommand{\aprxs}[1][]{$(1\pm\alpha#1)$-approximations\xspace}

\newcommand{\vt}[1]{v^{(#1)}}

\def\notes{0}
\ifnum\notes=1

\else
\fi

\title{Adversarially Robust Streaming via Dense--Sparse Trade-offs}
\author{Omri Ben-Eliezer\footnote{Work partially conducted while the author was at Harvard University.}\\MIT\\{\tt omrib@mit.edu} \and Talya Eden\\MIT\\{\tt talyaa01@gmail.com} \and Krzysztof Onak\\Boston University\\{\tt krzysztof@onak.pl}}
\date{}

\begin{document}
\maketitle 

\begin{abstract}
	A streaming algorithm is adversarially robust if it is guaranteed to perform correctly even in the presence of an adaptive adversary. The development and analysis of such algorithms have been a very active topic recently, and several sophisticated frameworks for robustification of classical streaming algorithms have been developed. One of the main open questions in this area is whether efficient adversarially robust algorithms exist for moment estimation problems (e.g., $F_2$-estimation) under the turnstile streaming model, where both insertions and deletions are allowed. So far, the best known space complexity for streams of length $m$, achieved using differential privacy (DP) based techniques, is of order $\tilde{O}(m^{1/2})$ for computing a constant-factor approximation with high constant probability (the $\tilde{O}$ notation hides here terms polynomial in $\log m$ and $\log n$, where $n$ is the universe size).
	In this work, we propose a new simple approach to tracking moments by alternating between two different regimes: a sparse regime, in which we can explicitly maintain the current frequency vector and use standard sparse recovery techniques, and a dense regime, in which we make use of existing DP-based robustification frameworks. The results obtained using our technique break the previous $m^{1/2}$ barrier for any fixed $p$. More specifically, our space complexity for $F_2$-estimation is $\tilde{O}(m^{2/5})$ and for $F_0$-estimation, i.e., counting the number of distinct elements, it is $\tilde O(m^{1/3})$.

All existing robustness frameworks have their space complexity depend multiplicatively on a parameter $\lambda$ called the \emph{flip number} of the streaming problem, where $\lambda = m$ in turnstile moment estimation. The best known dependence in these frameworks (for constant factor approximation) is of order $\tilde{O}(\lambda^{1/2})$, and it is known to be tight for certain problems. Again, our approach breaks this barrier, achieving a dependence of order $\tilde{O}(\lambda^{1/2 - c(p)})$ for $F_p$-estimation, where $c(p) > 0$ depends only on $p$.
\end{abstract}

\newpage
\section{Introduction}
Streaming algorithms are an integral part of the modern toolbox for large-scale data analysis. A~streaming algorithm observes a stream of data updates that arrive one by one, and is required to compute some global function of the data using a small amount of space (memory) and with an efficient running time. 

Most of the literature on streaming algorithms implicitly assumes that the stream updates do not depend on previous outputs of the algorithm or on the randomness produced by the algorithm. 
This assumption may not be realistic in many situations: for example, when the data is chosen by a malicious adversary in response to previous outputs, or when data characteristics change based on previous outcomes in some complicated or unpredictable way. 
As a result, the last couple of years have seen substantial progress in the systematic investigation of \emph{adversarially robust} streaming algorithms \cite{BY20, BJWY20, HassidimKMMS20, WoodruffZhou20, ABDMNY21, KMNS21, braverman2021adversarial, ACSS21}, which preserve their correctness guarantees even for adaptively chosen data and are thus especially suitable for these interactive settings.

There is already a wide range of problems and settings for which the best known adversarially robust streaming algorithms are almost as efficient as their classical, non-robust counterparts. 
The \emph{flip number} \cite{BJWY20} of a streaming problem, an algorithmic stability parameter which counts how many times the output value may change by a multiplicative factor of $1+\alpha$ as the stream progresses, plays a central role in many of these results \cite{HassidimKMMS20, WoodruffZhou20, KMNS21, ACSS21}. When the flip number $\lambda$ is small, the generic methods developed in these works can turn a classical streaming algorithm into an adversarially robust one with only a small overhead (linear in $\lambda$ or better) 
in the space complexity. This is especially useful in the \emph{insertion only} streaming model, where elements are only added to the stream, but may not be deleted from it. Many important streaming problems, such as $F_p$-estimation, distinct elements, entropy estimation, and various others, all have flip number of $\lambda = O(\alpha^{-1}\log m)$ for insertion-only streams of length $m$. Under the standard assumption that $m = \poly(n)$, where $n$ is the size of the universe of all possible data elements, and building on additional known results from the streaming literature, one can then obtain adversarially robust insertion-only \aprx algorithms with space complexity $\poly(1/\alpha, \log n)$.

The situation in the \emph{turnstile streaming} model, which allows both insertions and deletions, is more complicated. The most popular technique for turnstile streams in the classical regime, linear sketching, is provably not adversarially robust \cite{HardtWoodruff13}.
Furthermore, the flip number can be very large, potentially even   $\Theta(m)$.  The best known robustification methods in this regime \cite{HassidimKMMS20,ACSS21}, based on differential privacy, have a multiplicative $O(\sqrt{\lambda})$ dependence in the flip number (for constant $\eps$). Therefore,  
they induce a space overhead of $\tilde{\Omega}(\sqrt{m})$ compared to the best non-robust algorithms. 

A separation result of Kaplan, Mansour, Nissim and Stemmer~\cite{KMNS21} shows that indeed the $\sqrt{\lambda}$-type dependence in the flip number is tight for some streaming problems; specifically, they show this for a variant of the Adaptive Data Analysis problem in the context of bounded-space computation. 
We note, however, that the lower bound of \cite{KMNS21} does not apply to many core problems in the streaming literature, for which no separation between the classical oblivious and adversarially robust settings is known. In particular, this is the case for $F_p$-estimation, in which the goal is to approximate $\sum_i |v_i|^p$, the $p$-th moment of a frequency vector $v \in \Z^n$.
This gives rise to the following question, widely regarded as one of the central 
open questions on adversarially robust streaming.\footnote{%
To the best of our knowledge, the first explicit appearance of this question in 
the literature is in Jayaram's Ph.D.\ thesis \cite[page 26]{jayaram2021thesis}. 
See also a talk by Stemmer~\cite{stemmer_talk} at 54:45 and the third question 
on the list of open questions from the STOC 2021 Workshop on Robust Streaming, 
Sketching, and Sampling \cite{workshop}.}
\begin{center}
\makebox[\textwidth][s]
{\emph{What is the adversarially robust space complexity of $F_p$-estimation in the turnstile streaming model?}}
\end{center}

In this work we show that a combination of existing building blocks from the literature (with slight modifications and simplifications) can yield a substantially improved space complexity for the above problem. Our results hold when deletions are allowed, as long as each update increases or decreases the frequency of a single element by $1$ (or more generally, by a bounded integer amount). We also allow frequencies of elements to become negative in the process, which is known as the \emph{general turnstile streaming model}.

\section{Overview of Our Contribution}

\subsection{Our results}

We give an $F_p$-estimation algorithm that breaks the $\sqrt{m}$ (or $\sqrt{\lambda}$) space barrier. We now state a \emph{simplified} version of the main result, focusing just on on the dependence on the stream length $m$ and domain size $n$, whenever it is polynomial. 
For the full statement of our results, see Theorem~\ref{thm:main}.

\begin{theorem}[Simplified main result]\label{thm:main_sketch}
For any fixed $p \in [0,\infty)$ and $\alpha > 0$, there is an adversarially robust $F_p$-estimation streaming algorithm
that computes a \aprx, using:
\begin{itemize}
 \item $\tO(m^{1/3})$ space if $p \in [0,1]$,
 \item $\tO(m^{p/(2p+1)})$ if $p \in [1,2]$,
 \item $\tO(m^{p/(2p+1)} \cdot n^{1-5/(2p+1)})$ if $p \in (2,\infty)$,
\end{itemize}
where the $\tO$ notation suppresses factors that are polynomial in $\alpha^{-1}$, $\log m$, and $\log n$.
The algorithm gives correct estimates throughout the entire stream with probability $1-o(1)$.
\end{theorem}

We note that since the flip number for the moment estimation problem is $\lambda=\Theta(m)$ (see Section \ref{subsec:adv_robust_streaming}), the dependency of the space complexity of our approach in $\lambda$ is $\widetilde{O}(\lambda^{1/3})$ for $p\in[0,1]$ and $\widetilde{O}(\lambda^{\frac{1}{2}-\frac{1}{4p+2}})$ for $p>1$.
This improves polynomially upon the currently best known $\widetilde{O}(\sqrt \lambda)$ bound, obtained using the aforementioned differential privacy based robustness frameworks \cite{HassidimKMMS20,ACSS21}.
Together with the separation of \cite{KMNS21}, our result (see also \cite{jayaram2021thesis}) suggests that a paradigm shift may be required in order to achieve improved space complexity for turnstile streams: rather than developing very widely applicable robustness frameworks that suffer from the $\sqrt{\lambda}$-type lower bound due to their wide applicability, it may make sense to look for methods that are perhaps somewhat less generic, and exploit other properties of the problem, beyond just the flip number.

\newcommand{\mainAlg}{
\begin{algorithm}[t!]
\caption{Adversarially Robust Streaming of Moments\\
\algheader{Parameters} $p \in [0,\infty)$ describing the moment,
dimension $n\in\Z_+$,
stream length 
\noheader{Parameters} bound $m \in \Z_+$,
threshold $T \in \R_+$, approximation quality parameter\\
\noheader{Parameters} 
$\alpha \in (0,1)$, error parameter $\delta$}
\label{alg:main}

$\vregime \leftarrow \sparse$; $v \leftarrow (0,\ldots,0)$; $\vcount \leftarrow 0$\\
$\Mexact \leftarrow 0$; $\Mapprox \leftarrow 0$; $\kapprox \leftarrow 0$\\
$\Asparse \leftarrow {}$the sparse recovery algorithm (Theorem~\ref{thm:sparse}) with sparsity parameter $k=\lceil 4T \rceil$\\
$\vinterval \leftarrow \begin{cases}
                            \alpha T/4 & \mbox{for }p\in[0,1]\\
                            \frac{\alpha}{32p}\left(\frac{\alpha T}{16}\right)^{1/p} & \mbox{for }p \in (1,\infty)
                            \end{cases}$\\
$\vinterval \leftarrow \max\{\lfloor\vinterval\rfloor,1\}$\\
$\Adensity \leftarrow {}${}$(m/\lfloor T/10\rfloor)$-query adversarially robust streaming algorithm (Algorithm \ref{alg:bounded_query})\newline for $(1\pm.25)$-approximation of number of distinct elements with error parameter $\delta/2$\\
$\Aapprox \leftarrow {}${}$(m/\vinterval)$-query adversarially robust streaming algorithm (Algorithm \ref{alg:bounded_query})\newline for \aprx[/4] of the $p$-th moment with error parameter $\delta/2$\\
\ForEach{\rm update $(i,\Delta)$}{
$\vcount \leftarrow \vcount + 1$\\
Process the update $(i,\Delta)$ by $\Asparse$, $\Adensity$, $\Aapprox$\\
\lIf{\rm$\vcount$ is a multiple of $\lfloor T/10\rfloor$}{$\kapprox \leftarrow {}$estimate from $\Adensity$}
\lIf{\rm$\vcount$ is a multiple of $\vinterval$}{$\Mapprox \leftarrow {}$estimate from $\Aapprox$}
\If{\rm$\vregime=\sparse$}{
Update $v$ and $\Mexact$\\
Output $\Mexact$\\
\lIf{\rm$\moment{0}{v} \ge 4T$}{$\vregime \leftarrow \dense$}
}
\Else{
Output $\Mapprox$\\
\If{\rm$\kapprox \le 2T$}{
Use $\Asparse$ to recover $v$\\
$\Mexact \leftarrow \pmoment{v}$\\
$\vregime \leftarrow \sparse$
}}
}
\end{algorithm}
}
\subsection{Our techniques}

Our result relies on a straightforward combination of known techniques from 
the literature. The bottleneck of the previous best result for moment estimation for 
general turnstile streaming is the direct reliance on the flip number 
$\lambda$, which for general streams can be of order $\Omega(m)$. As 
mentioned above, methods that take only the flip number into account (and do not use any other characteristics of the problem at hand) cannot get space complexity much smaller than $\sqrt{\lambda}$ (that is, $\sqrt{m}$ for norm estimation). Thus, we exploit a specific characteristic of
$F_p$-estimation: the actual number of significant 
changes to the $p$-th moment can only be large if the moment remains small.
This can only be the case if the underlying vector is sparse, i.e., has 
relatively few non-zero coordinates. We therefore divide the current 
state of the frequency vector into two regimes: sparse and dense, using a 
threshold $T$. If the vector has at most $T$ non-zero coordinates, it is 
considered sparse. If it has more than $4T$ non-zero coordinates, it is 
considered dense. For densities in between, the state of the vector may be 
temporarily classified as either dense or sparse. 

In the sparse regime,
we take the simplest possible approach, which is storing the input explicitly, 
using a sparse representation, which requires only ${O}(T)$ space. In this form, 
it is easy to maintain the current moment, since we know the current frequency vector
exactly. In the dense regime, we apply the technique from the paper of Hassidim, 
Kaplan, Mansour, Matias, and Stemmer~\cite{HassidimKMMS20},
which uses differential
privacy to protect not the input data, but rather the internal randomness of estimators
it uses. 
At a high level, their framework consists of invoking a set of $k$ estimators, which upon each query provide an updated estimate. 
Given the stream of updates, they  use differentially private methods to detect whenever the current estimate is no longer relevant, at which point they query the set of estimators to get an updated estimate.
Their technique in general increases the space requirement by a factor of
the square root of the flip number, compared to that of oblivious streaming algorithms.
In particular, applying their method for the moment estimation problems, 
requires invoking  $k=\widetilde{O}(\sqrt \lambda)$  instances of  oblivious $F_p$-estimation algorithms.
We improve on the above,  by taking advantage of the fact that the estimated value of the $p$-th moment cannot change too rapidly for dense vectors. For instance, for $p=0$ (the distinct element
count) or $p=1$, if the vector has at least $T$ non-zero coordinates, at least $\Omega(T)$ insertions or deletions are required to change it by a constant factor. Similarly, for $p=2$, at least $\Omega(\sqrt T)$ insertions or deletions are needed. Hence, the flip number for the dense regime is much lower, and  we can make significantly fewer queries to a set of oblivious $F_p$-estimation algorithms. This in turn implies that we can significantly reduce the number of required estimators.

The missing component in our description so far is how the transition between the 
regimes happens. If we are transitioning from the sparse regime to the dense one, we 
have all the information needed about the current state of the input vector that 
we are tracking. If we are transitioning from the dense regime to the sparse 
regime, we use off-the-shelf sparse recovery techniques (also known as 
compressed sensing) to recover the frequency vector exactly. To know when to do 
this, we run in parallel an adversarially robust streaming algorithm for 
distinct element counting, which we also know has to be queried only every 
$\Omega(T)$ steps.

\mainAlg

\subsection{Pseudocode}
We present the pseudocode of our approach as 
Algorithm~\ref{alg:main}. We maintain 
adversarially robust estimators, $\Aapprox$ and $\Adensity$, that  are queried  significantly less frequently than $m$ times throughout the entire execution of the algorithm. We query them at 
regular intervals, knowing that their values cannot change too rapidly, when the 
vector is dense. We note that we do not use them when the vector is sparse, as in that regime their 
readouts may be inaccurate.

We present the pseudocode for an adversarially robust algorithm that has to 
answer only a limited number of queries in Algorithm~\ref{alg:bounded_query}. This 
algorithm is a simplified and adjusted version of the algorithm that appeared in 
the work of Hassidim et al.~\cite{HassidimKMMS20}. In the introduction of their paper, they note that constructing 
a bounded-query variant of their algorithm is possible, but
do not give any details beyond that. As this variant is crucial for our 
purposes, we present this construction in full detail for completeness.

\begin{algorithm}[t]
\caption{Bounded Query Adversarially Robust Streaming\\
\algheader{Parameters} number $q$ of queries, probability $\delta$ of error, bound $\tau$ on the size of the range
\noheader{Parameters} of possible values, desired approximation parameter $\alpha$\\
\algheader{Subroutine} oblivious \aprx[/3] streaming algorithm $\AA$ as described in\\
\noheader{Subroutine} the statement of Lemma~\ref{lem:query_bounded}
}
\label{alg:bounded_query}
$\eps \leftarrow 1/100$\label{line:eps_def}\\
$\delta' \leftarrow \eps\delta/(10q)$\label{line:delta_prime_def} \\
$\eps' \leftarrow \eps/\sqrt{8q\ln(1/\delta')}$\label{line:eps_prime_def}\\
$k \leftarrow \Theta(\eps'^{-1}\log\frac{2q\log 2\tau}{\alpha\delta})$\label{line:k_def}\\
Initialize $k$ independent instances $\AA_1$, \ldots, $\AA_k$ of $\AA$\\
\ForEach{\rm update $(i,\Delta)$}{
Process the update $(i,\Delta)$ by all of $\AA_1$, \ldots, $\AA_k$\\
\If{\rm the adversary is querying for the current output}{
\ForEach{$j \in [k]$}{
$\gamma_j \leftarrow \mbox{estimate from\ }\AA_j$\\
$\gamma'_j \leftarrow \gamma_j$\mbox{\ rounded up to a multiple of\ }$1+\frac{\alpha}{3}$ and truncated if not in $\set{0} \cup [1,\tau)$.
}
Output an $(\eps',0)$-DP estimate of the median of 
$\set{\gamma'_1,\ldots,\gamma'_k}$ as in Theorem~\ref{thm:dp_median}\newline
with parameters set to 
$\delta \leftarrow \frac{\delta}{2q}$ and $\eps\leftarrow\eps'$
}
}
\end{algorithm}
\section{Preliminaries}

\subsection{Basic notation and terms}

For any $k \in \N$, we write $[k]$ to denote $\set{1,\ldots,k}$, the set of $k$ smallest positive integers.

\begin{definition}[The $p\th$-moment]
The $p\th$-moment of a vector $v\in \R^n$ is $\pmoment{v}=\sum_{i=1}^n |v_i|^p$, 
for any $p \in [0,\infty)$. We interpret the $0\th$-moment as
$\|v\|^0_0 = |\set{i \in [k]: v_i \ne 0}|$, i.e., the number of 
non-zero coordinates of $v$, 
by assuming in this context that $0^0 = 0$ and $x^0 = 1$ for any $x \ne 0$.
\end{definition}

\begin{definition}[Vector density]\label{def:dense}
Let $k \in \R$ and $n \in \N$. We say that a vector $v \in \R^n$ is 
\emph{$k$-dense} if at least $k$ of its coordinates are non-zero, and \emph{$k$-sparse} if at most $k$ of them are non-zero.
\end{definition}

\begin{definition}[\aprx]
For any $Q,Q' \in [0,\infty)$ and $\alpha \in [0,1]$, we say that $Q'$ is a \emph{\aprx} to $Q$ if
\[(1-\alpha)Q \le Q' \le (1+\alpha)Q.\]
\end{definition}

\subsection{Streaming algorithms}

A \emph{streaming algorithm} receives, one by one, a stream of updates $u_1$, $u_2$, \ldots, 
$u_m$ that modify the data, and is typically required to compute or approximate 
some function $f$ of the data over
the stream of updates. In this paper, we fully focus on the setting in which the 
input is a \emph{frequency vector} $v \in \Z^n$ for some integer $n$, known to 
the algorithm in advance. Initially, at the beginning of the stream, this vector 
is the all-zero vector, i.e., $v=(0,\ldots,0)$. The stream consists of updates 
of the form $u_j = (i_j,\Delta_j)$ in which $i_j \in [n]$ and 
$\Delta_j \in \set{-1,1}$. The interpretation of each update is that $\Delta_j$ is 
added to the $i_j$-th coordinate of $v$, i.e., each update increases 
(``insertion'') or decreases (``deletion'') a select coordinate by 
1.\footnote{As mentioned, the update values $\Delta_j$ are always $\pm 1$ in the 
model we consider here. In the most general setting for turnstile streaming, 
$\Delta_j$ may be unbounded; however, for our arguments to hold, it is important 
that $\Delta_j$ are bounded in some way (e.g., satisfy $\Delta_j \in [-C, 
-C^{-1}] \cup \{0\} \cup [C^{-1}, C]$ for some constant $C \geq 1$).
This assumption is not necessary for $F_0$-estimation---i.e., counting the 
number of distinct elements---for which updates of arbitrary magnitude are 
allowed as long as they can be handled by a non-robust streaming algorithm on 
which we build.}

For any fixed stream of updates, the streaming algorithm is required to output 
$f(v)$ or a good approximation to $f(v)$ (e.g., a \aprx if $f(v)$ is a 
non-negative real number) after seeing the stream with probability $1-\delta$, 
for some parameter $\delta$. We refer to $1-\delta$ in this context, as \emph{success probability}.
We sometimes refer to streaming algorithms in this model, in which the stream is independent of the actions of the streaming algorithm, as \emph{oblivious} or \emph{non-robust} to distinguish them from adversarially robust streaming algorithms, which we design in this paper.

\paragraph{Intermediate approximations.}
We assume that the streaming algorithm does not know the exact length of the 
stream in advance, and only knows an upper bound on it. Because of that, we 
assume that the algorithm can be asked to output its approximation of 
$f(v)$ at any time throughout the stream. This is true for a large majority of 
streaming algorithms, and in particular, to the best of our knowledge, applies 
to all general moment streaming algorithms. For this type of streaming algorithm,
if its success probability is $1-\delta$, this means that it can output a 
desired approximation with probability at least $1-\delta$ at any fixed prefix 
of the stream.

\paragraph{Streaming related notation and assumptions.}
Throughout the paper, we consistently write $n$ to denote the dimension of the vector on 
which the streaming algorithm operates. We use $m$ to denote an upper bound 
on the length of the input and we assume that $m = O(\poly(n))$. We assume that machine words are large enough to represent $n$ and $m$, i.e., the number of bits in them is at least $\Omega(\log \max\{m,n\})$ and we express the complexity of algorithms in words.

\subsection{Adversarially robust streaming algorithms}
\label{subsec:adv_robust_streaming}

In this paper, we design streaming algorithms in the adversarially robust streaming model of Ben-Eliezer et al.~\cite{BJWY20}, which we now introduce in the context of computing a \aprx to values of a function $f$.

\begin{definition}[Adversarially robust streaming]
\label{def:adv_robust_streaming}
Fix a function $f:\mathbb{Z}^n \to [0,\infty)$ and let $\alpha > 0$. The robust streaming model is defined as a game between two players, \AdversaryPlayer and \AlgorithmPlayer, where $f$, $\alpha$, and the stream length $m$ are known to both players. 
In each round $j \in [m]$:
\begin{itemize}
\item First, \AdversaryPlayer picks $(i_j, \Delta_j)$ for $i_j \in [n]$ and $\Delta_j \in \{-1,1\}$ and sends them to \AlgorithmPlayer. The choice of $i_j$ and $\Delta_j$ may depend on all previous updates sent by \AdversaryPlayer as well as all previous outputs of \AlgorithmPlayer.
\item \AlgorithmPlayer outputs $y_j$, which is required to be a \aprx to $f(\vt{j})$, where $\vt{j}$ is the vector aggregating all updates so far (that is, $\vt{j}_i = \sum_{j' \leq j: i_{j'} = i} \Delta_{j'}$ for all $i \in [n]$). \AlgorithmPlayer sends $y_j$ to \AdversaryPlayer.
\end{itemize}
\AlgorithmPlayer's goal is to return correct outputs at all times. That is, $y_j$ is required to be a \aprx to $f(\vt{j})$ for all $j \in [m]$.
Conversely, \AdversaryPlayer's goal is to have \AlgorithmPlayer's output $y_j$ that is not a \aprx to $f(\vt{j})$ for some $j \in [m]$.

We say that a streaming algorithm is \emph{adversarially robust} and has success probability $1-\delta$ for some $\delta\in[0,1]$ if it can provide all correct outputs as \AlgorithmPlayer with probability at least $1-\delta$ for any \AdversaryPlayer.
\end{definition}

We now introduce the notion of an \emph{$q$-query} adversarially robust streaming algorithm, which has to provide no more than $q$ outputs.
\begin{definition}[$q$-query robust streaming algorithm]
Let $q \in \mathbb{N}$. A \emph{$q$-query adversarially robust streaming 
algorithm} is defined similarly to a standard adversarially robust streaming 
algorithm, with the following modification: 
\AdversaryPlayer may perform at most 
$q$ queries for outputs $y_j$ from \AlgorithmPlayer, and only receives the output $y_j$ in these time steps where queries are made. \AdversaryPlayer may pick these time steps adaptively as a function of all previous interactions.
We say that a $q$-query adversarially robust streaming algorithm has a probability of success $1-\delta$, for some $\delta\in[0,1]$, if for any \AdversaryPlayer that makes at most $q$ queries, with probability at least $1-\delta$, it correctly answers all of them.
\end{definition}

\paragraph{Note on the tracking property.}
Oblivious streaming algorithms are not required to have the ``tracking'' property, i.e., they have to provide a good approximation at the end of the stream (or at any fixed point), but their definition does not require any type of consistent behavior throughout the stream. This is required, however, for adversarially robust streaming algorithms. We build our adversarially robust streaming algorithms from oblivious streaming algorithms that do not have a tracking property.

\subsection{Frequency moment estimation}

The main focus of this paper is designing streaming algorithms for the problem of $F_p$-estimation, i.e., estimation of the $p\th$-moment, in the turnstile  streaming model.
To build our adversarially robust algorithms, we use classical, non-robust, turnstile streaming algorithms for $F_p$-estimation.

\begin{theorem}[Previously known $F_p$-estimation results]\label{thm:approx_alg}
There exist turnstile streaming algorithms that with probability of success $\frac{9}{10}$, return a \aprx to the $p\th$ moment of a frequency vector in $\Z^n$ on the stream of length $m$ with the following space complexity:
\begin{center}
\rm\renewcommand{\arraystretch}{1.4}
    \begin{tabular}{|c|c|c|}
    \hline
     \textbf{Value of $p$} & \textbf{Space} & \textbf{Reference}\\ \hline
     $p=0$ &
     $O(\alpha^{-2}\cdot\log n\cdot (\log(1/\alpha)+\loglog(m)))$ &
     \cite{kane2010optimal} \\ \hline
    $p\in(0,2]$ & 
    $O(\alpha^{-2} \cdot \log m)$ & \cite{kane2010exact} \\ \hline
    $p> 2$ & 
    $O\left(n^{1-2/p} \cdot \left(\alpha^{-2}+\alpha^{-4/p}\log n \right)\right)$& \cite{GW18}  \\ \hline
    \end{tabular}
\end{center}
\end{theorem}

\subsection{Flip number}

The flip number, defined in \cite{BJWY20}, plays a prominent role in many of the previous results on adversarially robust streaming. For completeness, we next provide its definition suited for our context.
\begin{definition}[Flip number]
Fix a function $f \colon \Z^n \to [0,\infty)$ and $\alpha > 0$. Let $u_1 = (i_1, \Delta_1)$, \ldots, $u_m = (i_m, \Delta_m)$ be a sequence of updates to some vector $v$ whose initial value is $\vt{0}$, and let $\vt{j}$ be the value of the vector after $j$ updates have been received. The \emph{flip number} $\lambda_\alpha(f, (u_1, \ldots, u_m))$ of $f$ with respect to the above sequence is the size $t$ of the largest subsequence $0 \leq j_1 < \ldots < j_t \leq m$ for which $f\left(\vt{j_l}\right)$ is not a \aprx of  $f\left(\vt{j_{l+1}}\right)$ for any $l=1,\ldots,t-1$. The flip number $\lambda_\alpha(f)$ of $f$ is the maximum of $\lambda_\alpha(f, (u_1, \ldots, u_m))$ over all possible choices of the sequence $u_1$, \ldots, $u_m$.
\end{definition}

It is easy to see that the flip number of $F_p$-estimation is $\Omega(m)$ for any $p$: indeed, consider the following pair of insertion-deletion updates $(i,1), (i,-1)$, repeated $m/2$ times. In such a stream, the value of the $p\th$ moment alternates $m$ times between $0$ and $1$.

\subsection{Sparse recovery}

In our algorithm, we use sparse recovery to reconstruct the current frequency vector 
when it becomes sparse, which is possible even if it was arbitrarily dense in the meantime. For an 
introduction to the topic of sparse recovery, see the survey of Gilbert and 
Indyk~\cite{sparse_recovery}. Here we use the following streaming 
subroutine introduced by Gilbert, Strauss, Tropp, and 
Vershynin~\cite{GSTV}.

\begin{theorem}[Sparse recovery \cite{GSTV}]\label{thm:sparse}
There is a streaming algorithm that takes a parameter $k$, operates on a vector 
$v\in\Z^n$, and has the following properties. It uses $O(k \polylog n)$ words of 
space and handles each coordinate update in $O(\polylog n)$ time. Whenever the 
input vector is $k$-sparse, the algorithm can reconstruct it exactly in $O(k 
\polylog n)$ time.

With probability $1-O(n^{-3})$, taken over the initial selection of randomness, the algorithm can correctly recover \emph{all} $k$-sparse vectors in all parts of the process (even when they are constructed in an adaptive manner).
\end{theorem}

\subsection{Differential privacy}
Differential privacy \cite{DMNS2006} is by now a standard formal notion of privacy for individual data items in large datasets. The formal definition is as follows.
\begin{definition}[Differential Privacy]
Let $A$ be a randomized algorithm operating on databases. $A$ is $(\eps, \delta)$-differentially private (in short $(\eps, \delta)$-DP) if for any two databases $S$ and $S'$ that differ on one row, and any event $T$, it holds that 
$$
\Pr[A(S) \in T] \leq e^\eps \cdot \Pr[A(S') \in T] + \delta
$$
\end{definition}
In the framework of Hassidim et al.~\cite{HassidimKMMS20}, which we use here, DP is used in a somewhat non-standard way to protect the internal randomness of instances of a static algorithm.

\section{Bounded Query Adversarially Robust Streaming}
\label{sec:bounded_query_streaming}

In this section, we present a $q$-query adversarially robust streaming algorithm for 
approximating a function $f \colon \Z^n \to \R$.
In the case that $q \ll m$, and for problems where the flip number is $\lambda = \Theta(m)$, such as turnstile $F_p$-estimation, it obtains significant gains in the space complexity compared to the general algorithm introduced by Hassidim et al.~\cite{HassidimKMMS20}.
The space overhead of the $q$-query robust 
algorithm over an oblivious streaming algorithm is roughly $\sqrt{q}$, 
independently of how much the function changes in the meantime. Informally, this 
is because the flip number of the output observed by the adversary decreases 
from a (worst case) $\Theta(m)$ factor to a $\Theta(q)$ one. The algorithm is a 
simplified and adjusted version of the algorithm of Hassidim et al.~\cite{HassidimKMMS20}. 
Their algorithm builds on two important primitives: a DP procedure for detecting when a set of functions exceeds a certain 
threshold, and a DP procedure for computing the median 
of a set of values. Their algorithm works by invoking the threshold detection 
procedure after each update, in order to detect whether the estimate of the 
computed function should be re-evaluated. If this is the case, then the median 
procedure is used to compute a private updated estimation. 
Compared to their algorithm, we do not need the first primitive, i.e., the differentially private threshold detection.
We only recompute a private median when 
the algorithm is replying to a query from the adversary.\footnote{We note that the private thresholds procedure is crucial for Hassidim et al.\ \cite{HassidimKMMS20} to improve their space overhead from roughly $\sqrt{m}$, which strongly depends on the stream length, to roughly $\sqrt{\lambda}$, which depends only on the flip number. In the problems we consider here, this is not essential as $\lambda = \Theta(m)$ anyway.}

\begin{lemma}[$q$-query adversarially robust algorithm]\label{lem:query_bounded}
Let $\alpha,\delta \in (0,1)$ and $q \in \Z_+$. Let $\AA$ be an oblivious 
streaming algorithm for computing a \aprx[/3] to a function $f:\Z^n\to 
[0,\infty)$ that uses $S$ space 
and is correct with probability $9/10$ when queried once. Additionally, let 
$\set{0} \cup [1,\tau]$ be the range of possible correct values of $f$ on the 
stream.

There is a $q$-query adversarially robust streaming algorithm, Algorithm~\ref{alg:bounded_query},  that uses 
$\chi \cdot S$ space
to provide a \aprx to $f$ with probability $1-\delta$, where \[\chi \eqdef 
O\left(\sqrt{q \log (2q/\delta)} \cdot 
\log\frac{2q\log2\tau}{\alpha\delta}\right).\]
\end{lemma}

\subsection{Tools from differential privacy}

In order to prove Lemma~\ref{lem:query_bounded}, we use the following set of tools from the differential privacy literature.
First, the following theorem allows for composing multiple applications of a DP mechanism.

\begin{theorem}[\cite{DworkRV10}]\label{thm:composition}
Let $\eps,\delta' \in (0,1]$ and let $\delta \in [0,1]$. An algorithm that allows 
for $q$ adaptive interactions with an $(\eps,\delta)$-DP 
mechanism is $(\eps',q\delta+\delta')$-DP for $\eps' \eqdef 
\sqrt{2q\ln(1/\delta')} \cdot \eps + 2q\eps^2$.
\end{theorem}

At the heart of the algorithm is a DP mechanism for 
computing a median of a set of items.
While sublinear-space algorithms for DP median estimation are known to exist \cite{ABC2021}, for our purposes it suffices to use a simple approach with near-linear space complexity.

\newcommand{\hasmed}{\cite[Theorem 2.6]{HassidimKMMS20}}
\begin{theorem}[\hasmed]\label{thm:dp_median}
For every $\eps,\delta \in (0,1)$, there exists an $(\eps,0)$-DP algorithm 
for databases $S \in X^{*}$ of size $\Omega\left(\frac{1}{\eps}\log\left(|X|/\delta\right)\right)$ that 
outputs an element $x \in X$ such that with probability at least $1-\delta$, there are at least $|S|/2 - \Gamma$ elements in $S$ that are bigger or equal to $x$ and at least $|S|/2 - \Gamma$ elements in $S$ that are smaller or equal to $x$, where $\Gamma \eqdef O\left(\frac{1}{\eps}\log\left(|X|/\delta\right)\right)$.
The algorithm uses $O(|S|)$ space.
\end{theorem}
We note that the original statement of the theorem did not mention the space complexity, but such space can be obtained using standard approaches in the DP literature, e.g., by applying the exponential mechanism with the Gumbel trick \cite{ABC2021,ExpMechanism}.

Finally, we use a known generalization theorem that shows that a deferentially 
private mechanism cannot heavily bias its computation on a sufficiently large 
set of random samples.

\begin{theorem}[\cite{BassilyNSSSU21,DworkFHPRR15}]\label{thm:generalization}
Let $\eps \in (0,1/3)$, $\delta \in (0,\eps/4)$, and $t \ge {\eps}^{-2}\log(2\eps/\delta)$. Let $\AA:X^t \to 2^X$ be an $(\eps,\delta)$-DP algorithm that operates on a database of size $t$ and outputs a predicate $h:X \to \set{0,1}$. Let $\DD$ be a distribution on $X$, let $S$ be a database containing $t$ elements drawn independently from $\DD$, and let $h \leftarrow \AA(S)$ be an output of $\AA$ on $S$. Then
\[\Pr_{S\sim\DD\atop h\leftarrow \AA(S)}
 \left[
 \left|\frac{1}{|S|}\sum_{x\in S}h(x) - \E_{x\sim\DD}[h(x)]\right| \ge 10\eps
 \right]
 \le
 \frac{\delta}{\eps}.
\]
\end{theorem}

\subsection{Proof of Lemma~\ref{lem:query_bounded}}

Recall that Algorithm~\ref{alg:bounded_query} runs multiple copies $\AA_1$, \ldots, $\AA_k$ of an oblivious streaming algorithm $\AA$, where each $\AA_i$ uses an independent random string $r_i\in \zo^{*}$, selected from the same distribution as the randomness of $\AA$.
Let $R \eqdef \set{r_1,\ldots,r_k}$ be the collection of random strings used by 
the copies of $\AA$. We can view $R$ as a database, in which each $r_i$ is a 
row, and Algorithm~\ref{alg:bounded_query} as a mechanism that operates on it.
We now show that Algorithm~\ref{alg:bounded_query} does not reveal much about 
the collection of random strings it uses.

\begin{lemma}\label{lem:privacy}
Algorithm~\ref{alg:bounded_query} 
is $(\eps,\delta')$-DP with respect to $R$, the collection of random strings used by copies of $\AA$, where $\eps$ and $\delta'$ are as defined in the algorithm.
\end{lemma}

\begin{proof} 
The only way in which the algorithm reveals anything about the strings $r_i$ is 
by outputting the private median of current estimates of all algorithms. Note 
that the set of possible values of estimates $\gamma'_j$ is of size at most 
$1+\lceil\log_{1+\alpha/3}(\tau)\rceil = O(\alpha^{-1}\log (2\tau))$. Let 
$\eps'$ be as defined in Line~\ref{line:eps_prime_def}. It follows from 
Theorem~\ref{thm:dp_median} that each application of the median algorithm is 
$(\eps',0)$-DP with respect to $R$ and errs with probability at most $\delta/(2q)$ when the constant 
hidden by the asymptotic notation in the definition of $k$ in Line~\ref{line:k_def} is large enough. Applying Theorem~\ref{thm:composition}, 
we conclude that the entire algorithm is $(\eps'',\delta)$-DP with respect to $R$, where \[\eps'' 
\eqdef \sqrt{2q\ln(1/\delta')} \cdot \eps' + 2q\eps'^2 \le \frac{\eps}{2} + 
\frac{\eps^2}{4}\le\eps.\qedhere\]
\end{proof}

We now proceed to prove Lemma~\ref{lem:query_bounded}, i.e., that Algorithm~\ref{alg:bounded_query} has the desired properties.

\begin{proof}[Proof of Lemma~\ref{lem:query_bounded}.]
Recall that by Lemma~\ref{lem:privacy}, Algorithm~\ref{alg:bounded_query} is 
$(\eps,\delta')$-DP with respect to the collection of random strings that copies of $\AA$ use, where $\eps$ is defined in 
Line~\ref{line:eps_def} and $\delta'$ is defined in Line~\ref{line:delta_prime_def}. 
For any random string $r\in\zo^{*}$ that an instance of $\AA$ may use and for any 
$i \in [q]$, let $h_i(r):\zo^{*} \to \zo$ equal $1$ if $\AA$ outputs a \aprx[/3] to the function being computed
on the prefix of the stream, when asked the $i$-th query and using $r$ as its randomness.
Let $r_j$ be the randomness that $\AA_j$, the $j$-th instance of $\AA$, uses.
Since the adversary can be seen as an $(\eps,\delta')$-DP mechanism, and this includes all generated queries and updates to the stream, by Theorem~\ref{thm:generalization} and the union bound,
we get that
\[\left|\E_r[h_i(r)] - \frac{1}{k}\sum_{j=1}^k h_i(r_j)\right|\le 
10\eps=\frac{1}{10}\]
for all $i\in[q]$, with probability at least $1 - q \cdot \frac{\delta'}{\eps} \ge 1-\delta/10$
as long as $k \ge \eps^{-2}\log(2\eps/\delta')$. We now show that this condition holds for a sufficiently 
large constant hidden by the asymptotic notation in the definition of $k$ in 
Line~\ref{line:k_def}. First, observe that $\eps$ is defined to be a positive constant, and hence $\eps^{-2}$ is a constant as well and can easily be bounded by a sufficiently large constant hidden in the definition of $k$. It remains to bound $\log(2\eps/\delta') = \log 
(q/(5\delta))$. To this end, observe that the definition of $k$ also has two multiplicative terms. The first one is 
\[\frac{1}{\eps'} = \frac{\sqrt{8q\ln(1/\delta')}}{\eps} \ge 100\sqrt{8 \ln 10} \ge 1,\]
and the second one is  
\[\log\frac{2q\log 2\tau}{\alpha\delta} \ge \log\frac{2q}{\delta} = \log 2 + \log\frac{q}{\delta}.\]
Since $\frac{q}{\delta} \ge 1$, and, therefore, $\log \frac{q}{\delta} \ge 0$, their product
multiplied by a sufficiently large constant---which again can be hidden in the asymptotic notation in the definition of $k$---is
greater than $\log\frac{20q}{\delta} = \log 20 + \log\frac{q}{\delta}$. This 
finishes the proof that $k \ge \eps^{-2}\log(2\eps/\delta')$ can easily be achieved by properly adjusting constants.

Since $\AA$ is correct with probability at least $9/10$ on any fixed data stream,
this implies that for each $i\in[q]$, at least $\frac{4}{5}k$ predicates $h_i(r_j)$ are $1$. In other words,  with probability at least $1-\delta/10$, for each query from the adversary,
at least $\frac{4}{5}k$ of the collected estimates $\gamma_j$ of the current value of $f$ are its \aprxs[/3].
Note that the rounding step can only increase the approximation error by a factor of at most $(1+\alpha/3)$, which means that at least $\frac{4}{5}k$ of estimates $\gamma'_j$ are \aprx of the current value of $f$ because $(1+\alpha/3)^2 < 1+\alpha$.

Now note that the private median algorithm returns an estimate that is  
greater than or equal to at least $2/5$ of estimates $\gamma'_j$ and also 
smaller than or equal to at least $2/5$ of the same estimates with probability at least $1-\frac{\delta}{2q}$.
As long as this algorithm outputs such an estimate, and 
as long as the fraction of bad estimates is at most $1/5$, this means that the algorithm 
outputs a \aprx to the value of $f$ at the query point.
By the union bound this occurs for all queries with probability at least
$1-\delta/2-\delta/10\ge 1-\delta$.

The space complexity of the algorithm is dominated by 
the space to store the instances of $\AA$. Note that there are 
$k = O(\sqrt{q \log (2q/\delta)} \cdot \log\frac{2q\log2\tau}{\alpha\delta})$ of them.
\end{proof}

\section{Bounded Change}

As discussed in the introduction, our frequency moment estimation algorithm gains from the fact that when the vector $v$ is $k$-dense for some value of $k$, the value of $\pmoment{v}$ cannot change too rapidly. We formally state and prove this below.

\subsection{Moments $p \in [0,1]$}

\begin{lemma}\label{lem:bc_p01}
Let $v,v'\in \Z^n$, $p \in [0,1]$, $\alpha \in (0,1)$, and $k \in \Z_{+}$.
If $v$ is $k$-dense and $\|v-v'\|_1 \le \alpha k$, then $\pmoment{v'}$ is a \aprx to $\pmoment{v}$.
\end{lemma}

\begin{proof}
Consider a function $f:\R \to \R$ defined as $f(t) \eqdef |t|^p$. We claim that for any $t \in \Z$, $|f(t) - f(t+1)| \le 
1$. This is easy to verify for $t \in \set{-1,0}$, because $f(-1) = f(1) = 1$ 
and $f(0)=0$. Since $f$ is differentiable in $(-\infty,-1]\cup[1,\infty)$ with 
the absolute value of the derivative bounded by $1$, the claim holds in that range as well, i.e., for other $t \in \Z \setminus \set{-1,0}$.
This implies that for any $t,t' \in \Z$, $|f(t) - f(t')| \le |t - t'|$.

We have
\[\left|\pmoment{v} - \pmoment{v'}\right| 
\le \sum_{i=1}^{n} \left| f(v_i) - f(v'_i) \right| 
\le \sum_{i=1}^{n}|v_i - v'_i|
= \|v - v'\|_1
\le \alpha k.
\]
Since $v$ is $k$-dense, for at least $k$ of its coordinates $i$, $|v_i| \ge 1$, and hence $\pmoment{v} \ge k$.
We therefore have $\pmoment{v'} = \pmoment{v} + \left(\pmoment{v'} - \pmoment{v}\right) \le \pmoment{v} + \alpha k \le (1+\alpha) \pmoment{v}$. Analogously, $\pmoment{v'} \ge \pmoment{v} - \alpha k \ge (1-\alpha)\pmoment{v}$.
\end{proof}

\subsection{Moments $p \in [1,\infty)$}

We use the following two well-known facts, which are easy to verify via basic calculus.

\begin{fact}\label{fact:ineq_lb}
For $p\ge 1$ and $\alpha \in (0,1)$, $(1-\alpha)^p \ge 1-\alpha p$.
\end{fact}

\begin{fact}\label{fact:ineq_ub}
For $\alpha \in (0,1)$, $e^\alpha \le (1+2\alpha)$.
\end{fact}

\begin{lemma}\label{lem:bc_p1inf}
Let $v,v'\in \Z^n$, $p \in [1,\infty)$, $\alpha \in (0,1)$, and $k \in \Z_{+}$.
If $v$ is $k$-dense and $\|v-v'\|_1 \le \frac{\alpha}{8p}(\frac{\alpha k}{4})^{1/p}$, then
$\pmoment{v'}$ is a \aprx to $\pmoment{v}$.
\end{lemma}

\begin{proof}
Let $\Delta \eqdef v-v'$. 
We partition the set of indices, $[n]$, into two sets, 
$\Ism$ and $\Ilg$, based on how $|\Delta_i|$ compares to $|v_i|$. 
We have $\Ism \eqdef \set{i \in [n] : |\Delta_i| \le \frac{\alpha}{4p}|v_i|}$ and $\Ilg 
\eqdef [n] \setminus \Ism$.

For $i \in \Ism$, we have \[ 
\left((1-\frac{\alpha}{4p}) |v_i|\right)^p 
\le 
|v'_i|^p
\le
\left((1+\frac{\alpha}{4p}) |v_i|\right)^p.
\]
The left--hand side can be bounded from below by $(1-\alpha/4)|v_i|^p$, using Fact~\ref{fact:ineq_lb}.
The right--hand side is at most $\left(e^{\alpha/4p}\right)^p |v_i|^p \le e^{\alpha/4} |v_i|^p \le
(1+\frac{\alpha}{2})|v_i|^p$, where the last inequality uses Fact~\ref{fact:ineq_ub}.
This implies that $\left||v_i|^p - |v'_i|^p\right| \le \frac{\alpha}{2} |v_i|^p$.
As a corollary, we obtain
\[\sum_{i \in \Ism} \left||v_i|^p - |v'_i|^p\right| \le \frac{\alpha}{2} \sum_{i \in \Ism} |v_i|^p
\le \frac{\alpha}{2}\pmoment{v}.
\]

For $i \in \Ilg$, we have 
\[
\sum_{i\in\Ilg}|v'_i|
\le
\sum_{i\in\Ilg}|v_i| + |\Delta_i|
\le
\sum_{i\in\Ilg}\left(1+\frac{4p}{\alpha}\right) |\Delta_i|
\le
\sum_{i\in\Ilg}\frac{8p}{\alpha} |\Delta_i|
\le
\frac{8p}{\alpha}\|\Delta\|_1
\le
\left(\frac{\alpha k}{4}\right)^{1/p}.
\]
This implies that, due to the convexity of the function $f(x) \eqdef x^p$ for $p\ge1$,
\[
\sum_{i\in\Ilg} |v'_i|^p
\le \left(\sum_{i \in \Ilg}|v'_i|\right)^p
\le \frac{\alpha k}{4}.
\]
The same bound holds for $v$, i.e.,
\[
\sum_{i\in\Ilg} |v_i|^p
\le \left(\sum_{i \in \Ilg}|v'_i|\right)^p
\le \frac{\alpha k}{4}.
\]
Combining these bounds, we get a bound on the sum of differences in coordinates in $\Ilg$
\[
\sum_{i\in\Ilg} \left||v_i|^p - |v'_i|^p\right|
\le
\sum_{i\in\Ilg} |v_i|^p + |v'_i|^p
\le \frac{\alpha k}{4} + \frac{\alpha k}{4}
= \frac{\alpha k}{2}
\le \frac{\alpha}{2}\pmoment{v},
\]
where the last inequality follows from the fact that $v$ is $k$-dense, and therefore, $\pmoment{v} \ge k$.

Overall, combining our knowledge for both $\Ism$ and $\Ilg$,
\begin{align*}
\left| \pmoment{v} - \pmoment{v'} \right|
&= \left| \sum_{i=1}^{n} |v_i|^p - |v'_i|^p \right|
\le \sum_{i=1}^{n} \left| |v_i|^p - |v'_i|^p \right| \\
&\le \sum_{i\in\Ism} \left| |v_i|^p - |v'_i|^p \right|
+ \sum_{i\in\Ilg} \left| |v_i|^p - |v'_i|^p \right| \\
&\le \alpha\pmoment{v}.
\end{align*}
This immediately implies our main claim.
\end{proof}

\section{Proof of the Main Result}

We start by restating our main result. This version has more details than the 
simplified version, which was presented as Theorem~\ref{thm:main_sketch} in the 
introduction.

\begin{theorem}[Adversarially robust moment estimation algorithm, full version of Theorem~\ref{thm:main_sketch}] \label{thm:main}
Algorithm~\ref{alg:main} is a \aprx adversarially robust streaming algorithm 
for the $p\th$-moment
with success probability $1-\delta-O(n^{-3})$ for streams of length $m$. 
The space complexity of the algorithm for different values of $p$ is specified in Table~\ref{table:our_results}.
\end{theorem}

\renewcommand{\arraystretch}{1.6}
\begin{table}[h!] 
    \centering
    \begin{tabular}{ | p{5em} | p{10em}| p{23.5em} | } 
    \hline
     \textbf{value of $p$} & \textbf{$\widetilde{O}(m^\mu n^\rho)$ space} & \textbf{detailed space complexity} \\ \hline
    $p\in[0,1]$ &  $\mu=\frac{1}{3}$, $\rho=0$ & 
    $O(m^{1/3}\cdot \alpha^{-5/3}\cdot \log^{5/3}(m/\alpha\delta))\cdot \polylog(n)$ \\ \hline
    \vspace{1pt} $p\in(1,2]$ &    
    \vspace{1pt} $\mu = \frac{p}{2p+1}$, $\rho = 0$ & 
    $O\left(m^{p/(2p+1)}\cdot \alpha^{-(5p+1)/(2p+1)}\cdot \log^{5p/(2p+1)}\left(m/(\alpha\delta)\right)\right)$ 
    $\cdot \polylog (n)$ \\ \hline
    $p=2$   &  $\mu=\frac{2}{5}$, $\rho=0$    &  $O\left(m^{2/5}\cdot \alpha^{-11/5} \log^{4/3}(m/\alpha\delta)\right)\cdot \polylog (n) $ \\ \hline
    \vspace{1pt}
    $p\in(2,\infty)$ & 
    \vspace{1pt}
    $\mu = \frac{p}{2p+1}$, $\rho = 1 - \frac{5}{2p+1}$ &
    $O\left((mp)^{p/(2p+1)} \cdot n^{1-5/(2p+1)}\cdot \alpha^{-(5p+1)/(2p+1)}\right)$
     $\cdot \log^{10p/(6p+3)}(m/(\alpha\delta))\cdot \polylog n$ \\ \hline
    \end{tabular}
    \caption{Space complexity of Algorithm~\ref{alg:main}. See Theorem~\ref{thm:main}.}\label{table:our_results}
\end{table}

\paragraph{Implementation notes for Algorithm~\ref{alg:main}.}\label{sec:imp_notes}
We start with a few implementation details. First, to ensure low memory usage,
one has to maintain a sparse 
representation of $v$, i.e., store only the non-zero coordinates in an easily 
searchable data structure such as a balanced binary search tree. This is 
possible, because as soon as $v$ becomes $4T$-dense, we stop maintaining 
it explicitly. Hence this part of the algorithm uses only $O(T)$ words of space.

We also avoid discussing numerical issues, and assume that for any integer 
$j \in [m]$, we can compute a good approximation to $j^p$ in $O(1)$ time, and 
also that summing such sufficiently good approximations still yields a 
sufficiently good approximation. In order to efficiently update $\Mexact$, while 
avoiding accumulating numerical errors (due to a sequence of additions and 
subtractions), one can create a balanced binary tree in which we sum 
approximations for $|v_i|^p$ for all non-zero coordinates $v_i$. Updating one of 
them then requires only updating the sums on the path to the root. This path is 
of length $O(\log T)$, and hence this requires updating at most 
$O(\log T)$ intermediate sums, each being a result of adding two values.

\paragraph{Proof of our main result.}
We are now ready to move on to the proof of our main result, which collects all the tools that we have developed throughout the paper.

\begin{proof}[Proof of Theorem~\ref{thm:main}]
We first prove our algorithm's correctness conditioning on three assumptions, and then we prove that these assumptions hold with high probability. Finally, we analyze the space complexity of the algorithm. Throughout the proof $\vt{i}$ denotes the value of $v$ after the $i\th$ update.
Our assumptions are:
\begin{enumerate}
    \item All invocations of $\Asparse$, $\Aapprox$, and $\Adensity$ are successful. That is, the following events occur: $\Asparse$ correctly recovers $v$ (provided that $v$ is $\lceil4T\rceil$-sparse), $\Aapprox$ returns a $(1\pm\alpha/4)$-approximation to the $p\th$-moment of $v$ whenever it is queried, and $\Adensity$ returns a $(1\pm.25)$-approximation to the number of non-zero coordinates in $v$ whenever it is queried.\label{assump:success}
    \item If $v$ is $T$-sparse, then $\vregime=\sparse$. \label{assump:ksparse}
    \item If $\vregime=\sparse$, then $v$ is $4T$-sparse. \label{assump:sparse}
\end{enumerate}

\textbf{Correctness under the assumptions.} By Assumption~\ref{assump:sparse}, $\Asparse$ is only invoked when $v$ is $4T$-sparse. By  the first item, each such invocation correctly recovers $v$. Hence, at the first time step of every time interval such that $\vregime=\sparse$, the algorithm has a sparse representation of $v$ (as discussed in Section~\ref{sec:imp_notes}) and this continues for the duration of the sparse interval. Therefore, for the duration of an interval where $\vregime=\sparse$, $\Mexact$ correctly approximates $\pmoment{v}$, and therefore all outputs of the algorithm are \aprxs to $\pmoment{v}$.

Consider now a time interval where $\vregime=\dense$. 
By the above discussion, at the first time step such that $\vregime=\dense$, it holds that $\pmoment{v}>4T$ (since during \sparse\ intervals the algorithm exactly knows  $\pmoment{v}$). 
We claim that at all time steps where $\vregime=\dense$, $\kapprox$ is a $(1\pm\alpha)$-approximation of $\pmoment{v}$. 
Fix a maximal time interval $[t,t']$  such that $\vregime=\dense$. 
Let $\Mapprox(t)$ denote the value of $\Mapprox$ at time step $t$, and let $i_0$ denote the time step in which this value was computed (note that $i_0\leq t$). 
Further let $i_1, \ldots, i_\ell$ denote all the time steps within $[t,t']$ in which $\Aapprox$ was invoked.
Now consider any two subsequent time steps $i_j, i_{j+1}$ for
$j\in[1,\ell-1]$, and any time step $z\in [i_{j}, i_{j+1}]$.
It holds that $i_{j+1}-i_j=\vinterval$. 
For $p\in[0,1],$ since $z-i_j\leq \vinterval=\alpha T/4$, it holds that $\|\vt{i_j}-\vt{z}\|_1\leq \alpha T/4$, and  by Lemma~\ref{lem:bc_p01} it follows that $\pmoment{\vt{z}}\in (1\pm\alpha/4)\pmoment{\vt{i_j}}$.
For $p\geq 1$,  $\|\vt{i_j}-\vt{z}\|_1\leq \vinterval=\frac{\alpha}{32p}(\frac{\alpha T}{16})^{1/p}$, and  Lemma~\ref{lem:bc_p1inf} implies that 
$\pmoment{\vt{z}}\in (1\pm\alpha/4)\pmoment{\vt{i_j}}$.
 By the assumption that all invocations of $\Aapprox$ are successful, $\Mapprox(i_j)\in(1\pm\alpha/4)\pmoment{v}$.
 Hence, it follows that for both possible regimes of $p$, $\Mapprox(z)\in (1\pm\alpha)\pmoment{\vt{z}}$  for any $z\in [i_1, t']$. 
Similar reasoning proves that at time step $i_0$, $v$ was $k$-dense, and hence for any $z\in [t, i_1]$, it holds holds that $\Mapprox(z)\in (1\pm\alpha)\pmoment{\vt{z}}$.
Therefore, for any $z\in [t,t']$ such that $[t,t']$ is a maximal time interval with $\vregime=\dense$, it holds that the output of the algorithm is a $(1\pm\alpha)$-approximation of $\pmoment{v}$.
Hence, it remains to prove that the assumptions hold with high probability.

\textbf{The assumptions hold.}
For item~\ref{assump:success},  by Theorem~\ref{thm:sparse}, with probability $1-O(n^{-3})$, $\Asparse$ is successful on \emph{all} invocations.\footnote{We note that with probability $1-O(n^{-3})$ (over its set of initial random coins) Algorithm $\Asparse$ correctly recovers \emph{all} $k$-sparse vectors, and hence its output is correct for any (adversarial) input stream.} 
Algorithm $\Adensity$ is queried $O(m/\lfloor T/10\rfloor)$ times, so by the setting of $q$ it holds that, with probability at least $1-\delta/2$, all queries return a $(1\pm.25)$-approximation of $\moment{0}{v}$.
Similarly, $\Aapprox$ is invoked  $O(m/\vinterval)$ times, and hence, with probability at least $1-\delta/2$, all queries return a $(1\pm\alpha/4)$-approximation of $\pmoment{v}$.
Hence, Assumption~\ref{assump:success} holds with probability at least $1-\delta-O(n^{-3})$. We henceforth condition on this event.

Now consider Assumption~\ref{assump:ksparse}, that if  $v$ is $T$-sparse then $\vregime=\sparse$. Clearly this holds from the beginning of the stream and until the first time that $\moment{0}{v}>4T$, since up to that point everything is deterministic and exact.
Assume towards contradiction that there exists  time steps  such that $v$ is $T$-sparse  and $\vregime=\dense$, and let $t$ be the earliest one.
Let $i_0$ be the closest step prior to $t$ in which $\kapprox$ was recomputed (by invoking $\Adensity$).  By the conditioning on Assumption~\ref{assump:success} holding, $\kapprox(i_0)\in 
\left[.75\moment{0}{\vt{i_0}},1.25\moment{0}{\vt{i_0}}\right]$. Since at time step $i_0$, the regime was not changed to \sparse, it also holds that 
$\kapprox(i_0)>2T$. Hence, $\moment{0}{\vt{i_0}}>(4/5)\cdot \kapprox(i_0)>(8/5)T$.  
Clearly, in $\lfloor T/10\rfloor$ updates $\moment{0}{v}$ cannot change by more than $T/10$. Hence, $\moment{0}{\vt{t}}>T$,  implying that $v$ is not $T$-sparse, and so we have reached a contradiction. 

We turn to Assumption~\ref{assump:sparse}. By the conditioning on Assumption~\ref{assump:success} holding, 
$\Asparse$ always correctly recovers $v$, implying that
as long as $\vregime=\sparse$, $v$ is exactly known to the algorithm. Therefore, it can be exactly detected when $\moment{0}{v}$ becomes greater than $4T$, at which point the regime is being set to \dense. Hence, up until that point, $\moment{0}{v}\leq 4T$ and $v$ is $4T$-sparse.

\textbf{Space complexity analysis.} We now analyze the space complexity of our algorithm.
By Theorem~\ref{thm:approx_alg} and
Lemma~\ref{lem:query_bounded}, for $q=m/\lfloor T/10 \rfloor$, Algorithm $\Adensity$ requires $O(\sqrt{m/T}\cdot \alpha^{-2}\cdot \log^{5/2}(m/(\alpha\delta))\cdot  \log n=\widetilde{O}(\sqrt{m/T})$ space. This complexity is always bounded by the following terms.
By Theorem~\ref{thm:sparse}, $\Asparse$ requires $O(k\polylog n)=O(T\polylog n)$ space. 
We continue to analyze the space complexity due to $\Aapprox$ separately for different regimes of $p$. 

For $p\in[0,1]$, by Theorem~\ref{thm:approx_alg} and Lemma~\ref{lem:query_bounded}, Algorithm $\Aapprox$ with $q=m/\vinterval=O(m/(\alpha T))$, requires $O(\sqrt{m/(\alpha T)}\cdot \log^{5/2}(m/(\alpha\delta)) \cdot \alpha^{-2}\cdot \log n =\widetilde{O}(\sqrt{m/T})$ space. 
Hence, setting $T$ to balance between the space complexities of $\Asparse$ and $\Aapprox$ (and $\Adensity$), 
we get that
the space complexity of $(\alpha,\delta)$-approximating the $p\th$ moment for $p\in[0,1]$ is
\[O\left(m^{1/3}\cdot \alpha^{-5/3}\cdot \log^{5/3}(m/(\alpha\delta) \right)\cdot \polylog(n)= 
\widetilde{O}(m^{1/3}).\]
For $p\in(1,\infty]$, it holds that  $\vinterval=\Theta((\alpha/p)\cdot (\alpha T)^{1/p})$. We again consider two separate regimes. First, $p\in(1,2]$.
By Lemma~\ref{lem:query_bounded} and Theorem~\ref{thm:approx_alg}, since $q=\sqrt{m/\vinterval}$,
$\Aapprox$ takes $O\left(\sqrt{mp/(\alpha(\alpha T)^{1/p})} \cdot \log^{3/2}(m/(\alpha \delta))\cdot  \alpha^{-2}\cdot \log m\cdot\log(1/\delta)\right) =\widetilde{O}(\sqrt{m/T^{1/p}})$ space.
Equating this term with the $O(T\cdot \polylog n)$ space required by $\Asparse$, 
results in a space complexity of 
\[
O\left(m^{p/(2p+1)}\cdot \alpha^{-(5p+1)/(2p+1)}\cdot \log^{5p/(2p+1)}\left(m/(\alpha\delta)\right)\right)\cdot \polylog (n)= \widetilde{O}(m^{p/(2p+1)})
\]
for $p\in(1,2]$.
Finally, we consider the regime $p\in(2,\infty).$
In this regime, we again have $q=(m/\vinterval)=O(mp/(\alpha(\alpha T)^{1/p}))$, and so by Lemma~\ref{lem:query_bounded} and Theorem~\ref{thm:approx_alg},  the space usage of $\Aapprox$ is
$O(\sqrt{mp/(\alpha(\alpha T)^{1/p})}\cdot \log^{3/2}(m/(\alpha\delta)) (\alpha^{-2}\cdot \log(1/\delta)\cdot n^{1-2/p}))=\widetilde{O}(\sqrt{m/T^{1/p}}\cdot n^{1-2/p})$.
Hence, equating this with the $O(T\polylog n)$ space required by $\Asparse$, we get 
 a space complexity of
\begin{align*}
&O\left((mp)^{p/(2p+1)} \cdot n^{1-5/(2p+1)}
    \alpha^{-(5p+1)/(2p+1)}\cdot \log^{5p/(2p+1)}\left(m/(\alpha\delta)\right)\right)\cdot \polylog n  \\
    & = \widetilde{O}((mp)^{p/(2p+1)} \cdot n^{1-5/(2p+1)}).
\end{align*}
 This concludes the proof.
\end{proof}

\section*{Acknowledgments}

This work was inspired by the conversation of Cameron Musco and David Woodruff---after David Woodruff's talk at the STOC 2021 workshop on adversarially robust streaming \cite{workshop,workshop_recording}---about when tracking moments in the general 
turnstile model is difficult and involves a large flip number. The authors wish to thank Rajesh Jayaram and Uri Stemmer for useful discussions.

\bibliographystyle{alpha}
\bibliography{references}

\end{document}